\documentclass[11pt]{article}

\usepackage{epsfig,cancel,amsmath,amsthm,amssymb}
\usepackage{color}
\usepackage{hyperref}
\usepackage{overpic}
\usepackage{comment}
\newcommand{\ud}{\,\mathrm{d}}

\newtheorem{theorem}{Theorem}[section]
\newtheorem{lemma}[theorem]{Lemma}
\newtheorem{proposition}[theorem]{Proposition}

\newtheorem{corollary}[theorem]{Corollary}

\theoremstyle{definition}
\newtheorem{remark}{Remark}[section]

\def\ds{\displaystyle}

\DeclareMathOperator*{\Res}{Res}
\DeclareMathOperator*{\diag}{diag}

\DeclareMathOperator{\re}{Re}
\DeclareMathOperator{\im}{Im}
\renewcommand{\Re}{\re}
\renewcommand{\Im}{\im}

\numberwithin{equation}{section}

\title{Singular values of products of Ginibre random matrices, multiple orthogonal polynomials
and hard edge scaling limits}
\author{Arno B.J. Kuijlaars\footnotemark[1] ~and Lun Zhang\footnotemark[2]}
\date{\today}

\begin{document}

\maketitle
\renewcommand{\thefootnote}{\fnsymbol{footnote}}
\footnotetext[1]{Department of Mathematics, KU Leuven,
Celestijnenlaan 200B, B-3001 Leuven, Belgium. E-mail:
arno.kuijlaars\symbol{'100}wis.kuleuven.be}
\footnotetext[2] {School of Mathematical Sciences and Shanghai Key Laboratory for 
Contemporary Applied Mathematics, Fudan University, Shanghai 200433, People's Republic of China. E-mail:
lunzhang\symbol{'100}fudan.edu.cn }

\begin{abstract}
Akemann, Ipsen and Kieburg recently showed that the squared
singular values of products of $M$ rectangular random matrices with independent
complex Gaussian entries are distributed according to a determinantal
point process with a correlation kernel that can be expressed in
terms of Meijer G-functions. We show that this point process can be interpreted as a
multiple orthogonal polynomial ensemble. We give integral representations
for the relevant multiple orthogonal polynomials and a new double
contour integral for the correlation kernel, which allows us to find
its scaling limits at the origin (hard edge).
The limiting kernels generalize the classical Bessel kernels. For $M=2$ they
coincide with the scaling limits found by Bertola, Gekhtman, and Szmigielski
in the Cauchy-Laguerre two-matrix model, which indicates that these kernels
represent a new universality class in random matrix theory.
\end{abstract}

\section{Introduction}
\subsection{Products of Ginibre random matrices}

Random matrix theory is a broad field with many applications 
in mathematics, physics, and beyond, as is witnessed by the
survey volume \cite{ABF} and the recent monographs \cite{AGZ,Dei,ForBook,Tao}. 
Of particular importance for the development of the theory 
is the connection with determinantal point processes. Whenever 
the eigenvalues of a random matrix ensemble are a
determinantal point process, one has explicit expressions for the
eigenvalue distributions in terms of the correlation kernel. Tools from integrable systems may then be 
used to further analyze the correlation kernel in the large $n$ limit, in order to 
establish, for example, universality of local eigenvalue correlations.
It is a  recent discovery that products of random matrices can
fall in the framework of determinantal point processes.

The topic of products of  random matrices can be traced back to 
the work of Furstenberg and Kesten \cite{FK}, where the interest lies 
in the asymptotic behavior as the number of factors in the product tends to infinity.
This work has been highly influential with important applications 
in Schr\"{o}dinger operator theory \cite{BL} and in statistical physics 
relating to disordered and chaotic dynamical systems \cite{CPV}.

A more recent development is the study of eigenvalue and singular value
distributions for the products of random matrices with a fixed number of factors, but
allowing the size of the matrices to tend to infinity. With tools from free probability and diagrammatic expansions,  
one may find the limiting global eigenvalue distributions as in \cite{BBCC, BJW, BJLNS, PZ}.
It turns out that, as in the theory of a single random matrix, the various limits exhibit 
a rich and interesting  mathematical structure, which also show a large degree of universality, see e.g.\ \cite{GT, OS}.  
Apart from physical applications, the study is also motivated by other fields like MIMO (multiple-input and multiple-output)
networks in telecommunication \cite{TV}.

Akemann and Burda \cite{AB} proved that the eigenvalues of products of complex Ginibre 
matrices are determinantal in the complex plane, see \cite{Ip} for an extension
to quaternionic Ginibre matrices. 
A similar determinantal structure holds for the eigenvalues of products of truncated unitary matrices
\cite{ABKN}. The determinantal structure opens up the way to a more detailed analysis
at the finite $n$ level \cite{ABKN, AS}.
Very recently, Akemann, Kieburg, and Wei \cite{AKW} found that the squared singular values of products of
complex Ginibre matrices are a determinantal point process on the positive real line. This was
further extended to the case of products of rectangular Ginibre matrices by Akemann, Ipsen and Kieburg \cite{AIK}.  
The correlation kernels in \cite{AB,ABKN,AIK,AKW,Ip} are all expressed in terms of Meijer G-functions.

In this paper we follow \cite{AIK}. We take $M \geq 1$ and
let $X_1, X_2,\ldots,X_M$ be complex random matrices whose entries
are independent with a complex Gaussian distribution, also known as
Ginibre random matrices. We assume $X_j$ has size $N_{j} \times
N_{j-1}$ and form the product
\begin{equation} \label{Ym}
Y_M = X_M X_{M-1} \cdots X_1.
\end{equation}
Our interest lies in the squared singular values of $Y_M$, that is
the eigenvalues of $Y_M^*Y_M$, where the superscript $^\ast $ stands
for conjugate transpose. We assume $N_0 = \min \{N_0, \ldots, N_M \}$, and write
\begin{equation} \label{nuj}
\nu_j = N_j - N_0, \qquad j=0, \ldots, M, \qquad n = N_0.
\end{equation}
Thus $\nu_0=0$ and $Y_M^* Y_M$ is a square matrix of size $n$.

The case for the products of square matrices (i.e., $\nu_j = 0$ for
every $j$) was considered by Akemann, Kieburg and Wei \cite{AKW},
who showed that the squared singular values are distributed
according to a determinantal point process with a correlation kernel
that can be expressed in terms of Meijer G-functions. This was
extended by Akemann, Ipsen and Kieburg \cite{AIK} to the general
rectangular case. The determinantal point process is a biorthogonal
ensemble \cite{Bor} with joint probability density function (see \cite[formula
(18)]{AIK})
\begin{equation} \label{jpdf}
    P(x_1, \ldots, x_n) =  \frac{1}{Z_n}  \prod_{j < k} (x_k-x_j)\,
        \det \left[ w_{k-1}(x_j) \right]_{j,k=1, \ldots, n},
    \end{equation}
where $x_j>0$, $j=1,\ldots,n$, are the squared singular values of
$Y_M$,
\begin{equation} \label{wk}
    w_k(x) = \mathop{{G^{{M,0}}_{{0,M}}}\/}\nolimits\!\left({- \atop \nu_M, \nu_{M-1},  \ldots, \nu_2, \nu_1 +
    k} \Big{|} x\right),
    \end{equation}
and normalization constant (see \cite[formula (21)]{AIK})
\[ Z_n = n!\prod_{i=1}^{n}\prod_{j=0}^M \Gamma(i+\nu_j).  \]
The function $w_k$ is a Meijer G-function (see e.g.\ \cite{BeSz,Luke} and
the Appendix for a brief introduction) which can be written as a Mellin-Barnes integral
\begin{equation} \label{wkasMB}
    w_k(x) =  \frac{1}{2\pi i} \int_{c-i\infty}^{c+i\infty} \Gamma(s+\nu_1 + k) \prod_{j=2}^{M} \Gamma(s+\nu_j)  x^{-s} \ud s,
    \qquad k=0, 1, \ldots,  \end{equation}
with $c > 0$. By the inversion formula for the Mellin transform we have
\begin{equation} \label{moments}
    \int_0^{\infty} w_k(x) x^{s-1} \ud x =  \Gamma(s+\nu_1 + k) \prod_{j=2}^{M} \Gamma(s+\nu_j), \qquad s > 0,
    \end{equation}
which in particular shows that the moments of $w_k$ are given as
products of Gamma functions.

By \eqref{wkasMB} and the functional equation of the Gamma function
$\Gamma(z+1)=z\Gamma(z)$, we have
\[ w_k(x) =
\frac{1}{2\pi i} \int_{c-i\infty}^{c+i\infty} (s+ \nu_1)_k \prod_{j=1}^{M} \Gamma(s+\nu_j)  x^{-s} \ud s, \]
where  the Pochhammer symbol
\[ (s+\nu_1)_k =\frac{\Gamma(s+\nu_1+k)}{\Gamma(s+\nu_1)} =(s+\nu_1)(s+\nu_1+1) \cdots (s+\nu_1 + k-1) \]
is a polynomial of degree $k$ in the variable $s$.
Then by taking linear combinations of the weights we could alternatively take
\begin{align}  \label{wktilde}
    \widetilde{w}_k(x) & = \frac{1}{2\pi i} \int_{c-i\infty}^{c+i\infty} s^k \prod_{j=1}^{M} \Gamma(s+\nu_j) x^{-s} \ud s,
    \end{align}
in the definition of \eqref{jpdf}. This representation shows that
\eqref{jpdf} is fully symmetric in all parameters $\nu_1, \ldots,
\nu_M$. Note that
\[ \widetilde{w}_k(x) = \left( - x \frac{\ud}{\ud x} \right)^k w_0(x), \]
which can be easily obtained from \eqref{wkasMB}.

\subsection{Biorthogonal functions and the correlation kernel}
From general properties of biorthogonal ensembles \cite{Bor}, it is known that \eqref{jpdf}
is a determinantal point process with correlation kernel
\begin{equation} \label{Kn}
    K_n(x,y) = \sum_{j=0}^{n-1} \sum_{k=0}^{n-1} x^j (M_n^{-1})_{k,j}
    w_k(y),
    \end{equation}
where $M_n$ is the matrix of moments of size $n \times n$,
\begin{equation} \label{Mmoment}
    M_n = \begin{pmatrix} \ds \int_0^{\infty} x^j w_k(x) \ud x \end{pmatrix}_{j,k=0, \ldots, n-1}.
    \end{equation}
In addition  we have
\begin{equation} \label{def:Kn}
    K_n(x,y) = \sum_{k=0}^{n-1} P_k(x) Q_k(y),
    \end{equation}
where for each $k = 0, 1, \ldots$, $P_k$ is a monic polynomial of degree $k$ and $Q_k$ belongs to
the linear span of $w_0, \ldots, w_k$ in such a way that
\begin{equation} \label{PkQkbio}
    \int_0^{\infty} P_j(x) Q_k(x) \ud x = \delta_{j,k}.
    \end{equation}
Thus the $P_k$ and $Q_k$ are biorthogonal functions that we consider
for every non-negative integer $k$, not just for $k \leq n-1$.

Akemann et al.\ \cite{AIK, AKW} studied an extension of
\eqref{jpdf} to a two-matrix model and obtained in this framework
that for certain polynomials $\widetilde{Q}_k$,
\begin{equation} \label{2MMduality}
    \int_0^{\infty} \int_0^{\infty}  P_j(x) \widetilde{Q}_k(y) w^M_{\nu} (x,y) \ud x \ud y = h_j^{M}\delta_{j,k},
    \end{equation}
with
\[ w^M_{\nu}(x,y) = y^{\nu_1-1} e^{-y}
    \mathop{{G^{{M-1,0}}_{{0,M-1}}}\/}\nolimits\!\left({- \atop \nu_M, \nu_{M-1},  \ldots, \nu_2} \Big{|} \frac{x}{y}\right) \]
and
\[h_j^M=\prod_{m=0}^M(j+\nu_m)!;\]
see \cite[formulas (25), (27) and (37)]{AIK}. We emphasize that
$\widetilde{Q}_k \neq Q_k$, since indeed $Q_k$ is not a polynomial
and $\widetilde{Q}_k$ is a multiple of the Laguerre polynomial
$L^{(\nu_1)}_k$; see \cite[formula (42)]{AIK}. The
biorthogonality \eqref{2MMduality} is related to \eqref{PkQkbio},
since
\[ Q_k(x) = \frac{1}{h_k^M} \int_0^{\infty} \widetilde{Q}_k(y) w^M_{\nu}(x,y) \ud y, \]
but we will not use this fact.

The starting point of this paper is the biorthogonality
\eqref{PkQkbio} and we first show that the polynomials $P_k$ can be
characterized as multiple orthogonal polynomials with respect to the
first $M$ weight functions $w_0, \ldots, w_{M-1}$. Hence, the point
process \eqref{jpdf} is a multiple orthogonal polynomial (MOP)
ensemble in the sense of \cite{Kui1,Kui2}. This further implies a
representation of the correlation kernel $K_n$ \eqref{def:Kn} in
terms of the associated Riemann-Hilbert problem, which is helpful
for future asymptotic analysis.


In \cite{AIK} it is shown that the biorthogonal functions $P_k$
and $Q_k$ have integral representations as Meijer G-functions. We
rederive these results in Section \ref{sec:integral representation}
using only the biorthogonality \eqref{PkQkbio}. The recurrence
relations of the biorthogonal functions are explicitly given in
Section \ref{sec:recurrence rel}.
We turn to the study of the function $K_n$ in Section \ref{sec:study
of Kn}. We derive a double contour integral representation of $K_n$, which allows
us to find its scaling limit at the origin (hard edge). The limiting
kernels generalize the classical Bessel kernel, and if $M=2$, it
coincides with the limiting kernels in the Cauchy-Laguerre
two-matrix model recently studied by Bertola, Gekhtman and
Szmigielski in \cite{BeGeSz}. Universality suggests that the new
limiting kernels should apply to more general situations for the
products of independent complex random matrices, thus, representing a
new universality class. Finally, we present the integrable form of
the limiting kernels in the sense of Its-Izergin-Korepin-Slavnov
\cite{IIKS}. For convenience of the reader, we include a short
introduction to the Meijer G-function in the Appendix.

\begin{remark}
It is possible to consider the probability density function \eqref{jpdf} for
general parameters $\nu_1, \ldots, \nu_M > -1$. The condition $\nu_j > -1$ is needed
in order to guarantee the existence of the moments in \eqref{Mmoment}. All the
constructions in this paper go through in that more general case.

However, we do not have a proof that \eqref{jpdf} is a probability density function
in the case of non-integer parameters, in particular we do not know that
\eqref{jpdf} is non-negative for all $x_1, \ldots, x_n$, although
we strongly suspect that it will be the case.
\end{remark}

\section{Multiple orthogonal polynomial ensemble}\label{sec:MOP ensemble}


\subsection{Multiple orthogonality}
Our first result is that the point process \eqref{jpdf} is a MOP
ensemble \cite{Kui1,Kui2} with $M$ weight functions $w_0, \ldots, w_{M-1}$, where the
$w_k$ are defined in \eqref{wk}. This follows from
the following lemma.

\begin{lemma}\label{lem:mop}
The linear span of the functions $w_0, w_1, \ldots, w_{n-1}$ is
equal to the linear span of the functions
\begin{equation} \label{MOPbasis}
    x \mapsto x^j w_k(x), \qquad k =0, \ldots, M-1, \quad k +jM < n.
    \end{equation}
\end{lemma}
\begin{proof}
The linear span of $w_0, w_1, \ldots, w_{n-1}$ consists of
all functions that can be written as
\begin{equation} \label{MOPspace}
    x \mapsto \frac{1}{2\pi i} \int_{c-i\infty}^{c+i\infty} q(s) \prod_{j=1}^M \Gamma(s+\nu_j) x^{-s} \ud s,
    \qquad \deg q(s) \leq n-1.
    \end{equation}
We have by \eqref{wkasMB} and a change of variables $s \mapsto s+j$,
\begin{align*}
x^j  w_k(x) & = \frac{1}{2\pi i} \int_{c-i\infty}^{c+i\infty} (s+\nu_1)_k  \prod_{l=1}^M \Gamma(s+\nu_l) x^{j-s} \ud s \\
& = \frac{1}{2\pi i} \int_{c-i\infty}^{c+i\infty} (s+\nu_1+j)_k
\prod_{l=1}^{M} (s+\nu_l)_j \prod_{l=1}^M \Gamma(s+\nu_l) x^{-s} \ud
s.
\end{align*}
This is of the form \eqref{MOPspace} with polynomial
\[ q(s) = (s+\nu_1+j)_k \prod_{l=1}^{M} (s+\nu_l)_j \]
of degree $k+jM$. Thus the functions \eqref{MOPbasis} belong to the
linear span of $w_0, \ldots, w_{n-1}$. It is readily seen that these
are independent since they correspond to polynomials $q(s)$ that
have different degrees.
\end{proof}

The polynomials $P_k$ are therefore MOPs of type II with respect to
the weights $w_0, \ldots, w_{M-1}$ and diagonal multiple indices,
i.e.,
\begin{equation*}
\int_0^{\infty}P_n(x)x^j w_k(x) \ud x=0, \qquad
    j=0,\ldots, \lceil \tfrac{n-k}{M}\rceil-1, \quad k=0,\ldots,M-1,
\end{equation*}
where $\lceil x \rceil$ denotes the smallest integer $\geq x$; see
\cite{IsBook,VAGK}.


\subsection{Riemann-Hilbert problem}
As a consequence of Lemma \ref{lem:mop}, the polynomial $P_n$ is
characterized by the following Riemann-Hilbert problem. We look for a
$(M+1)\times (M+1)$ matrix-valued function $Y : \mathbb C \setminus
[0, \infty) \to \mathbb C^{(M+1) \times (M+1)}$ that is analytic
with jump condition
\begin{equation} \label{Yjump}
    Y_+(x) = Y_-(x) \begin{pmatrix}
1 & w_0(x) & \cdots & w_{M-1}(x) \\
0 & 1 &  \cdots & 0 \\
\vdots & \vdots & \ddots & \vdots \\
0 & 0 & \cdots & 1 \end{pmatrix}, \qquad x\in  (0, \infty),
\end{equation} where
$Y_+$ ($Y_-$) denotes the limiting value from the upper (lower)
half-plane. As $z\to \infty$, we require
\begin{equation} \label{Yasymp}
    Y(z) = (I + O(1/z)) \diag \begin{pmatrix} z^n & z^{-n_0} & \cdots & z^{-n_{M-1}} \end{pmatrix},
    \end{equation}
where $n_k =  \lceil \frac{n-k}{M} \rceil$. Combined with appropriate local
conditions near the origin that depend on the parameters
$\nu_1,\nu_2,\ldots,\nu_M$, the Riemann-Hilbert problem \eqref{Yjump}--\eqref{Yasymp}
 has a unique solution and the $(1,1)$ entry of $Y$ is $P_{n}$; see \cite{VAGK}.
Also, one has
\begin{multline} \label{KnandRHP}
 K_n(x,y) =\\
    \frac{1}{2\pi i(x-y)}
    \begin{pmatrix} 0 & w_0(y) & \cdots & w_{M-1}(y) \end{pmatrix}
    Y_+^{-1}(y) Y_+(x) \begin{pmatrix} 1 \\ 0 \\ \vdots \\ 0
    \end{pmatrix},
\end{multline}
which is a manifestation of the Christoffel-Darboux formula for
multiple orthogonal polynomials; see \cite{DaKu}. The representation
\eqref{KnandRHP} is potentially useful for asymptotic analysis
although we will not pursue this here.

\subsection{Special case $M=2$}

We now take a look at the case $M=2$. If $M=2$, then
\[ w_0(x) = \frac{1}{2\pi i} \int_{c-i\infty}^{c+i \infty}
    \Gamma(s+\nu_1) \Gamma(s+\nu_2) x^{-s} \ud s. \]
This can be expressed in terms of the modified Bessel function of
second kind (a.k.a.\ the Macdonald function). The formula 10.32.13
of \cite{DLMF} says that
\begin{equation*}\label{eq:Knu integral}
2 K_{\nu}(2 \sqrt{x}) =
            \frac{x^{\nu/2}}{2\pi i} \int_{c-i\infty}^{c+i \infty}
                \Gamma(s) \Gamma(s-\nu) x^{-s} \ud s, \qquad
                c>\max(\nu, 0),
\end{equation*}
which after a change of variables $s \mapsto s + \nu + \alpha $ leads to
\[ 2 K_{\nu}(2 \sqrt{x}) =
            \frac{x^{-\nu/2-\alpha}}{2\pi i} \int_{c-i\infty}^{c+i \infty}
                \Gamma(s+\nu+\alpha) \Gamma(s+\alpha) x^{-s} \ud s, \quad c>\max(-\nu-\alpha,-\alpha). \]
We take $\alpha = \nu_2$, $\nu=\nu_1 -\nu_2$,
to find that
\begin{equation} \label{Knu1}
    w_0(x) = 2 x^{(\nu_1 + \nu_2)/2} K_{\nu_1 - \nu_2}(2 \sqrt{x}).
    \end{equation}
It will be convenient to assume that $\nu_1 \geq \nu_2$, which we
can do without loss of generality.

Similarly,
\begin{align} \nonumber
    w_1(x) & = \frac{1}{2\pi i} \int_{c-i\infty}^{c+i \infty}
    \Gamma(s+\nu_1+1) \Gamma(s+\nu_2) x^{-s} \ud s \\
    & = 2 x^{(\nu_1 + \nu_2+1)/2}
    K_{\nu_1 + 1 - \nu_2}(2 \sqrt{x}). \label{Knu2}
        \end{align}
Thus if $\rho_{\nu}(x) = 2 x^{\nu/2} K_{\nu}(2\sqrt{x})$, we have
\[ w_0(x) = x^{\alpha} \rho_{\nu}(x), \quad
w_1(x) = x^{\alpha} \rho_{\nu+1}(x). \]
Multiple orthogonal
polynomials associated with the two weights \eqref{Knu1}--\eqref{Knu2} were considered by Van
Assche and Yakubovich \cite{VAYa} for which they obtained four term
recurrence relations; see also \cite{CoCoVA} and \cite{ZhRo} for
asymptotic results for these polynomials. In the random matrix
context (i.e., the case where $\nu_j = N_j -N_0$ are integers), we
have
\[ \nu = N_1 - N_2, \qquad \alpha = N_2 - N_0. \]
For the special case $\nu=\alpha=0$ (i.e., the products of two
square matrices), this relation was first observed in \cite{Zhang}.

For general $M$, there is an $M+2$ term recurrence relation (this
follows from general theory of MOP, cf. \cite[Section~23.1.4]{IsBook}) and we
will determine the recurrence coefficients
explicitly in Section \ref{sec:recurrence rel}.


\section{Integral representations}\label{sec:integral representation}

Integral representations for the biorthogonal polynomials $P_k$ and
their dual functions $Q_k$ are given in \cite{AIK} where they were
derived from a two matrix model. We rederive these results directly from
the biorthogonality \eqref{PkQkbio}.

\subsection{Integral representation for $Q_k$}
Recall the biorthogonality \eqref{PkQkbio}. The biorthogonal function $Q_k$ has
the form
\[ Q_k(x) = \frac{1}{2\pi i} \int_{c-i\infty}^{c+i\infty} q_k(s)
\prod_{j=1}^M \Gamma(s+\nu_j) x^{-s} \ud s, \] where $q_k$ is a
polynomial of degree $k$. The biorthogonality \eqref{PkQkbio} then
says that
\[ \frac{1}{2\pi i} \int_0^{\infty}  \int_{c-i \infty}^{c+i\infty} P_l(x) q_k(s)
    \prod_{j=1}^{M} \Gamma(s+\nu_j)  x^{-s} \ud s \ud x = \delta_{l,k}. \]

It turns out that we can write down $q_k$ explicitly as stated in
the following proposition.
\begin{proposition}
We have
\begin{align} \label{QkRodr}
    Q_k(x) = \frac{(-1)^k}{\prod_{j=0}^M \Gamma(k+1+\nu_j)} \left( \frac{\ud}{\ud x} \right)^k \left( x^k w_0(x)
    \right),
    \end{align}
    and
\begin{align} \label{qk}
    q_k(s) = \frac{(s-k)_k}{\prod_{j=0}^M \Gamma(k+1+\nu_j)}.
    \end{align}
\end{proposition}
\begin{proof}
It is easy to see after applying an integration by parts $k$ times that
\[ \int_0^{\infty} x^l \left( \frac{\ud}{\ud x} \right)^k (x^k w_0(x))  \ud x = 0, \qquad \text{for } l < k. \]
Note that integrated terms do not contribute, since
\begin{equation} \label{w0:at0}
    w_0(x) = O(x^{\alpha} (\log x)^{r-1}), \qquad \text{ as } x \to 0+,
    \end{equation}
with $\alpha = \min(\nu_1, \ldots, \nu_M) > -1$ and $r = \# \{ j
\mid \nu_j = \alpha\}$, which can be deduced from properties of the
Mellin transform \eqref{moments}; see e.g.\  \cite[Theorem
4]{FlGoDu}, and since for $x\to+\infty$, we have
\[ w_0(x) = O\left(x^{\theta}e^{-Mx^{1/M}} \right), \qquad
\theta =\frac{1}{M} \left(\tfrac{1}{2}(1-M) + \sum_{j=1}^M \nu_j
\right); \] see \cite[Theorem 5.7.5]{Luke}.

Similarly,
\begin{align*}
    \int_0^{\infty} x^k \left( \frac{\ud}{\ud x} \right)^k (x^k w_0(x)) \ud x & =
    (-1)^k k! \int_0^{\infty}  x^k w_0(x) \ud x \\
    & = (-1)^k \prod_{j=0}^M \Gamma(k+1+\nu_j),
    \end{align*}
where we recall \eqref{moments}  and the fact that $\nu_0 = 0$.
Thus if $Q_k$ is defined by \eqref{QkRodr}, then we have
\begin{equation} \label{Qkortho}
    \int_0^{\infty} x^l Q_k(x) \ud x = \delta_{l,k}, \qquad \text{for } l = 0, 1, \ldots, k.
    \end{equation}

Since
\[ x^k w_0(x) = \frac{1}{2\pi i} \int_{c-i\infty}^{c+i\infty}
\prod_{j=0}^M \Gamma(s+\nu_j) x^{k-s} \ud s, \] we find by taking
$k$ derivatives that
\[ \left( \frac{\ud}{\ud x} \right)^k \left( x^k w_0(x) \right)
    = \frac{1}{2\pi i} \int_{c-i\infty}^{c+i\infty}
    (-1)^k (s-k)_k \prod_{j=0}^M \Gamma(s+\nu_j) x^{-s} \ud s. \]
Thus
\begin{equation} \label{Qkintegral0}
    Q_k(x) = \frac{1}{2\pi i} \int_{c-i\infty}^{c+i\infty} q_k(s) \prod_{j=0}^M \Gamma(s+\nu_j) x^{-s} \ud s,
    \end{equation}
with $q_k$ as in \eqref{qk}. This proves that $Q_k$ belongs to the linear span of $w_0, \ldots, w_{k-1}$
and \eqref{Qkortho} shows that it is indeed the biorthogonal function.
\end{proof}

Note that \eqref{QkRodr} is a Rodrigues-type formula for $Q_k$. Note also
that \eqref{Qkintegral0} is an integral representation, which because of \eqref{qk}
we may also write as
\begin{align} \label{Qkintegral}
    Q_k(x) & = \frac{1}{2\pi i \prod_{j=0}^M \Gamma(k+\nu_j+1)}
    \int_{c-i\infty}^{c+i\infty} \frac{\prod_{j=0}^M \Gamma(s+\nu_j)}{\Gamma(s-k)} x^{-s} \ud
    s.
    \end{align}
By \eqref{def:Meijer}, we can identify \eqref{Qkintegral} as a Meijer G-function:
\begin{align} \label{QkMeijerG}
  Q_k(x) =
    \frac{1}{\prod_{j=0}^M \Gamma(k+\nu_j+1)}
    \mathop{{G^{{M+1,0}}_{{1,M+1}}}\/}\nolimits\!\left({-k \atop \nu_0, \nu_{1},  \ldots, \nu_M
    } \Big{|} x \right).
    \end{align}
Up to a multiplicative constant and an easy transformation of the
Meijer G-function, \eqref{QkMeijerG} is the same as \cite[formula
(49)]{AIK}.

\subsection{Integral representation for $P_n$}
There is a similar integral representation for $P_n$.
\begin{proposition}
We have  for $x > 0$,
\begin{equation} \label{Pnintegral}
    P_n(x) =  \frac{\prod_{j=0}^M \Gamma(n+\nu_j+1)}{2\pi i} \oint_{\Sigma}
    \frac{\Gamma(t-n)}{\prod_{j=0}^M \Gamma(t+\nu_j+1)} x^t  \ud t,
    \end{equation}
where $\Sigma$ is a closed contour that encircles $0, 1, \ldots, n$ once in the positive direction.
\end{proposition}

\begin{proof} In the proof we assume that $P_n$ is given by \eqref{Pnintegral} and
we show that $P_n$ is a monic polynomial of degree $n$ satisfying
\begin{align} \label{Pnbio}
    \int_0^{\infty} P_n(x) \widetilde{w}_k(x) \ud x = 0, \qquad k=0, \ldots,
    n-1,
    \end{align}
where $\widetilde{w}_k$ is defined in \eqref{wktilde}.

The integrand in the right-hand side of \eqref{Pnintegral} is
meromorphic on $\mathbb C$ with simple poles at $0, 1, \ldots, n$
(the poles of the numerator at the negative integers are cancelled
by the poles of the factor $\Gamma(t+1)$ in the denominator). Thus
by the residue theorem
\[ P_n(x) = \prod_{j=0}^M \Gamma(n+\nu_j+1) \sum_{l=0}^n \Res_{t=l}
    \left( \frac{\Gamma(t-n)}{\prod_{j=0}^M \Gamma(t+\nu_j+1)} \right) x^l. \]
We can evaluate the residues to obtain
\begin{equation} \label{Pnhyper}
    P_n(x) =  \sum_{l=0}^n  \frac{(-1)^{n-l}}{(n-l)!}
    \frac{\prod_{j=0}^M \Gamma(n+\nu_j+1)}{\prod_{j=0}^M \Gamma(l+\nu_j+1)} x^l,
        \end{equation}
which shows that $P_n$ is a monic polynomial of degree $n$.

To verify \eqref{Pnbio} we use
\[ \int_0^{\infty} x^t \widetilde{w}_k(x) \ud x = (t+1)^k \prod_{j=1}^{M} \Gamma(t+\nu_j+1), \]
which follows from \eqref{wktilde} and the inversion formula for Mellin transforms.
Then we can compute by \eqref{Pnintegral} and an interchange of integrals,
\begin{align*}
    &\int_0^{\infty} P_n(x) \widetilde{w}_k(x) \ud x \\
    & =
    \frac{\prod_{j=0}^M \Gamma(n+\nu_j+1)}{2\pi i} \oint_{\Sigma}
        \frac{\Gamma(t-n)}{\prod_{j=0}^M \Gamma(t+\nu_j+1)} (t+1)^k \prod_{j=1}^{M} \Gamma(t+\nu_j +1) \ud t \\
        & = \frac{\prod_{j=0}^M \Gamma(n+\nu_j+1)}{2\pi i} \oint_{\Sigma}
        \frac{\Gamma(t-n) (t+1)^k}{\Gamma(t+1)}  \ud t \\
        & = \frac{\prod_{j=0}^M \Gamma(n+\nu_j+1)}{2\pi i} \oint_{\Sigma}
        \frac{(t+1)^k}{t(t-1) \cdots (t-n)}  \ud t.
    \end{align*}
The  remaining integrand is a rational function that behaves like $O(t^{k-n-1})$ as $t \to \infty$.
The contour $\Sigma$ encircles all the poles once in the positive direction. Thus by moving the contour to infinity,
we find that the integral vanishes for $k \leq n-1$, which is the required biorthogonality \eqref{Pnbio}.
\end{proof}

The formula \eqref{Pnhyper} shows that $P_n$ is a hypergeometric polynomial
\[ P_n(x) =  (-1)^n \prod_{j=1}^M \frac{\Gamma(n+\nu_j+1)}{\Gamma(\nu_j+1)}
    {\; }_1 F_M \left({-n \atop 1+ \nu_1, \ldots, 1+\nu_M} \Big{|} x \right), \]
    as in \cite[formula (44)]{AIK}.
We can also identify $P_n$ in \eqref{Pnintegral} as a Meijer G-function:
\begin{align} \label{PnMeijerG}
  P_n(x) =
    -\prod_{j=0}^M\Gamma(n+\nu_j+1)\mathop{{G^{{0,1}}_{{1,M+1}}}\/}\nolimits\!\left({n+1
\atop -\nu_0, -\nu_{1},  \ldots, -\nu_{M-1}, -\nu_{M}}\Big{|}
x\right),
    \end{align}
which is equivalent to \cite[formula (45)]{AIK}.

\section{Recurrence relations}\label{sec:recurrence rel}
By Lemma \ref{lem:mop} and general theory of MOPs (cf. \cite[Chapter
23]{IsBook}), it follows that the polynomials $P_n$ satisfy an $M+2$
term recurrence relation
\begin{equation} \label{Pnrecurrence}
    x P_n(x) = P_{n+1}(x) + \sum_{k=0}^M a_{k,n} P_{n-k}(x).
\end{equation}
There is a dual recurrence relation
\begin{equation} \label{Qnrecurrence}
    x Q_n(x) = Q_{n-1}(x) + \sum_{k=0}^M b_{k,n} Q_{n+k}(x),
    \end{equation}
where because of the biorthogonality \eqref{PkQkbio},
\begin{align*}
    a_{k,n} = \int_0^{\infty} x P_n(x) \, Q_{n-k}(x) \, \ud x, \qquad
  b_{k,n}  = \int_0^{\infty} P_{n+k}(x) \, x Q_n(x) \, \ud x.
    \end{align*}
Therefore
\begin{equation} \label{akandbk}
    a_{k,n}  = b_{k, n-k}.
    \end{equation}
It is the aim of this section to calculate these recurrence
coefficients explicitly.

\subsection{Coefficients $b_{k,n}$}
\begin{proposition}
We have for $k = 0, \ldots, M$,
\begin{equation} \label{bkn}
b_{k,n} = \left( \prod_{j=0}^M (n+\nu_j+1)_k \right)
        \sum_{ j= 0}^{k+1} \frac{(-1)^{k+1-j}}{j! (k+1-j)!}
            \prod_{i=0}^M (n +j + \nu_i).
            \end{equation}
\end{proposition}

\begin{proof}
We have from \eqref{Qkintegral0}, after a change of variable $s
\mapsto s+1$,
\begin{align*}
    x Q_n(x) & = \frac{1}{2\pi i} \int_{c-i\infty}^{c+i\infty}
    q_n(s) \prod_{j=1}^M \Gamma(s+\nu_j) x^{-s+1} \ud s \\
    & = \frac{1}{2\pi i} \int_{c-i\infty}^{c+i\infty}
    q_n(s+1) \prod_{j=1}^M \Gamma(s+\nu_j+1) x^{-s} \ud s \\
    & = \frac{1}{2\pi i} \int_{c-i\infty}^{c+i\infty}
    q_n(s+1) \prod_{j=1}^M (s+\nu_j) \prod_{j=1}^M \Gamma(s+\nu_j) x^{-s} \ud s.
    \end{align*}
Then $q_n(s+1) \prod_{j=1}^M (s+\nu_j)$ is a polynomial in $s$ of
degree $n+M$ and it is our task to show that
\begin{equation} \label{qnshift}
    q_n(s+1) \prod_{j=1}^M (s+\nu_j)
    = q_{n-1}(s) + \sum_{k=0}^M b_{k,n} q_{n+k}(s)
\end{equation}
with $b_{k,n}$ given by \eqref{bkn}.

By \eqref{qk} we have that all terms in \eqref{qnshift} are zero for
$s=1, \ldots, n-1$, i.e., all terms are divisible by $q_{n-1}(s)$.
If we do this division and use \eqref{qk} then we find that we have
to prove
\[
    \prod_{j=0}^M \frac{s+\nu_j}{n+\nu_j} =  1 + \sum_{k=0}^M  \frac{b_{k,n}}{\prod_{j=0}^M (n+\nu_j)_{k+1}} (s-n-k)_{k+1}.
    \]
Write $s=t+n$. Then we have to prove
\begin{equation} \label{bkntoprove}
    f(t) = \prod_{j=0}^M (n+\nu_j) + \sum_{k=0}^M
    \frac{b_{k,n}}{\prod_{j=0}^M (n+\nu_j+1)_k} (t-k)_{k+1},
    \end{equation}
as an identity for polynomials in $t$, where
\begin{equation} \label{ft} f(t) = \prod_{j=0}^M (t+n + \nu_j).
\end{equation}
 Both sides of \eqref{bkntoprove} have degree $M+1$ and for $t=0$ the identity \eqref{bkntoprove} is valid.
The polynomials $t \mapsto (t-k)_{k+1}$ for $k=0, \ldots, m$ are a
basis for the vector space of polynomials of degree $\leq M+1$ that
vanish at $t=0$. Then it is clear that there exists coefficients
$b_{k,n}$ such that \eqref{bkntoprove} holds.

By contour integration we obtain from \eqref{bkntoprove}
\begin{equation} \label{bkncontour}
    \frac{b_{k,n}}{\prod_{j=0}^M (n+\nu_j+1)_k} = \frac{1}{2\pi i} \oint_{\Sigma} \frac{f(t)}{(t-k-1)_{k+2}}
    \ud t, \qquad k = 0, \ldots, M, \end{equation}
where $\Sigma$ is a closed contour that encircles the points $0,
\ldots, k$ once in the positive direction. This leads by the residue
theorem to
\begin{align*}
    b_{k,n} & = \left(\prod_{j=0}^M  (n+\nu_j+1)_k \right)
        \sum_{j=0}^{k+1} (-1)^{k+1-j} \frac{f(j)}{j! (k+1-j)!},
\end{align*}
which gives \eqref{bkn} in view of the definition \eqref{ft} of
$f(t)$.
\end{proof}

\subsection{Coefficients $a_{k,n}$}

Because of \eqref{akandbk} we immediately find an expression for the
recurrence coefficients $a_{k,n}$.
\begin{corollary}
We have for $k=0, \ldots, M$,
\begin{align} \label{akn}
    a_{k,n} & = \left(\prod_{j=0}^M  (n-k+\nu_j+1)_k \right)
        \sum_{j=0}^{k+1} (-1)^{k+1-j} \frac{\prod_{i=0}^M (n-k+j+\nu_i)}{j! (k+1-j)!}.
\end{align}
Reversing the order of summation we also have
\[ a_{k,n} =
\left(\prod_{j=0}^M  (n-k+\nu_j+1)_k \right)
        \sum_{j=0}^{k+1} (-1)^{j} \frac{\prod_{i=0}^M (n+1-j+\nu_i)}{j! (k+1-j)!}. \]
\end{corollary}
\begin{proof} Use \eqref{akandbk}, \eqref{bkn} and reverse the order of summation.
\end{proof}

From \eqref{akn} we see that $a_{k,n}$ is a polynomial expression in
$n$, which seems to be of degree $k(M+1) + M+1 =  (k+1)(M+1)$.
However there is a cancellation in the leading order terms and
$a_{k,n}$ is actually a polynomial in $n$ of degree $(k+1)M$.

\begin{lemma}
For every $k$ we have
\[ a_{k,n} = \binom{M+1}{k+1}  n^{(k+1)M} + O\left(n^{(k+1)M-1}\right). \]
\end{lemma}

\begin{proof}
From \eqref{akandbk} and the contour integral representation
\eqref{bkncontour}  for $b_{k,n}$ we find
\begin{equation} \label{aknintegral}
    a_{k,n} = \left( \prod_{j=0}^M (n-k+\nu_j+1)_k  \right)
    \frac{1}{2\pi i} \oint_{\Sigma} \frac{g_n(t)}{(t-k-1)_{k+2}}
    \ud t,
    \end{equation}
where
\[ g_n(t) = \prod_{j=0}^M (t+n-k+\nu_j) = \sum_{l=0}^{M+1} p_l(t) n^l \]
is a polynomial of degree $M+1$ in $n$. The coefficient $p_l(t)$ is
a polynomial in $t$ of degree $\deg p_l(t) = M+1-l$. Thus
\[
    \frac{1}{2\pi i} \oint_{\Sigma} \frac{g_n(t)}{(t-k-1)_{k+2}}
    \ud t  =
    \sum_{l=0}^{M+1} \left( \frac{1}{2\pi i} \oint_{\Sigma}
        \frac{p_l(t)}{(t-k-1)_{k+2}} \ud t \right)   n^l. \]
The integral vanishes if $p_l$ is a polynomial of degree $\leq k$
since in that case the integrand is $O(t^{-2})$, and we can move the
contour to infinity. This happens for $l \geq M-k+1$. For $l = M-k$,
we have
\[ p_{M-k}(t) = \binom{M+1}{k+1} t^{k+1}  + O(t^k), \qquad \text{ as } t \to \infty, \]
and by a residue calculation at infinity we obtain
\[ \frac{1}{2\pi i} \oint_{\Sigma} \frac{p_{M-k}(t)}{(t-k-1)_{k+2}} \ud t = \binom{M+1}{k+1}. \]
Thus the second factor in the right-hand side of \eqref{aknintegral}
is a polynomial of degree $M-k$ in $n$ with leading coefficient
$\binom{M+1}{k+1}$.

The other factor is a monic polynomial in $n$ of degree $k(M+1)$.
Thus $a_{k,n}$ has degree $k(M+1) + M-k = (k+1)M$ with leading
coefficient $\binom{M+1}{k+1}$, as claimed in the lemma.
\end{proof}

Let's finally write down \eqref{aknintegral} for small values of
$M$.

\paragraph{Case $M=1$} For $M=1$ we have a three term recurrence
\[ xP_n(x) = P_{n+1}(x) + a_{0,n}P_n(x) + a_{1,n}P_{n-1}(x) \]
with
\[ a_{0,n} = 2n + \nu_1 + 1, \qquad a_{1,n} = n(n+\nu_1). \]
This is the recurrence relation for monic Laguerre polynomials with
parameter $\nu_1$.

\paragraph{Case $M=2$} For $M=2$ we have a four term recurrence
\[ xP_n(x) = P_{n+1}(x) + a_{0,n}P_n(x) + a_{1,n}P_{n-1}(x) + a_{2,n} P_{n-2}(x) \]
with
\begin{align*}
    a_{0,n} & = 3n^2+(3+ 2 \nu_1 + 2 \nu_2) n + (1 + \nu_1 + \nu_2 + \nu_1 \nu_2), \\
  a_{1,n} & = n(n+\nu_1)(n+\nu_2)(3n + \nu_1 +\nu_2), \\
    a_{2,n} & = n(n-1) (n+\nu_1)(n+\nu_1-1)(n+\nu_2)(n+\nu_2-1).
    \end{align*}
This agrees with the recurrence coefficients given in \cite[Theorem
4]{VAYa} if we use $\alpha = \nu_2$, $\nu = \nu_1 - \nu_2$.

\section{Double integral representation and large $n$ limit of $K_n$}
\label{sec:study of Kn}

In this section, we are concerned with the correlation kernel
$K_n(x,y)$ defined in \eqref{def:Kn}.

\subsection{Double integral formula for $K_n$}

The correlation kernel admits a double contour integral representation.
\begin{proposition} \label{prop:Knintegral} We have
\begin{equation} \label{Knintegral}
    K_n(x,y) =  \frac{1}{(2\pi i)^2} \int_{-1/2-i\infty}^{-1/2+i\infty} \ud s \oint_{\Sigma}  \ud t
        \prod_{j=0}^M   \frac{\Gamma(s+\nu_j+1)}{\Gamma(t+\nu_j+ 1)}
            \frac{\Gamma(t-n+1)}{\Gamma(s-n+1)}
        \frac{x^t y^{-s-1}}{s-t},
\end{equation}
where $\Sigma$ is a closed contour going around $0, 1, \ldots, n$ in
the positive direction and  $\Re t > -1/2$ for $t \in \Sigma$.
\end{proposition}
\begin{proof}
The correlation kernel \eqref{def:Kn} can be written  as a double
integral
\begin{equation} \label{Knintegral0}
    K_n(x,y) = \frac{1}{(2\pi i)^2} \int_{c-i\infty}^{c+i\infty} \ud s \oint_{\Sigma}  \ud t
    \prod_{j=0}^M   \frac{\Gamma(s+\nu_j)}{\Gamma(t+\nu_j+ 1)}
\sum_{k=0}^{n-1} \frac{\Gamma(t-k)}{\Gamma(s-k)} x^t y^{-s},
    \end{equation}
where we used the integral representation  \eqref{Pnintegral} for $P_k$
and  \eqref{Qkintegral} for $Q_k$.
From the functional equation $\Gamma(z+1)=z\Gamma(z)$, one can
easily check that
\[ (s-t-1) \frac{\Gamma(t-k)}{\Gamma(s-k)} =
\frac{\Gamma(t-k)}{\Gamma(s-k-1)} -
\frac{\Gamma(t-k+1)}{\Gamma(s-k)}, \] which means that there is a
telescoping sum
\begin{equation} \label{Gammatelescope}
    (s-t-1) \sum_{k=0}^{n-1} \frac{\Gamma(t-k)}{\Gamma(s-k)}
    = \frac{\Gamma(t-n+1)}{\Gamma(s-n)} -\frac{\Gamma(t+1)}{\Gamma(s)}.
        \end{equation}

We are going to make sure that $s-t-1 \neq 0$ when $s \in c +
i\mathbb R$ and $t \in \Sigma$. We do this by taking $c = 1/2$ and
let $\Sigma$ go around $0, 1, \ldots, n$ but with $\Re t > -1/2$ for
$t \in \Sigma$. Then we insert \eqref{Gammatelescope} into
\eqref{Knintegral0} and get
\begin{align*}
    K_n(x,y) =&  \frac{1}{(2\pi i)^2} \int_{1/2-i\infty}^{1/2+i\infty} \ud s \oint_{\Sigma}  \ud t
    \prod_{j=0}^M   \frac{\Gamma(s+\nu_j)}{\Gamma(t+\nu_j+ 1)}
            \frac{\Gamma(t-n+1)}{\Gamma(s-n)}       \frac{x^t y^{-s}}{s-t-1}
\nonumber \\ &- \frac{1}{(2\pi i)^2}
\int_{1/2-i\infty}^{1/2+i\infty} \ud s \oint_{\Sigma}  \ud t
    \prod_{j=0}^M   \frac{\Gamma(s+\nu_j)}{\Gamma(t+\nu_j+ 1)}
            \frac{x^t y^{-s}}{s-t-1}.
    \end{align*}
The $t$-integral in the second double integral vanishes by Cauchy's
theorem, since the integrand does not have any singularities inside
$\Sigma$. We change $s \mapsto s+1$ in the first double integral and
we obtain \eqref{Knintegral}.
\end{proof}

We can rewrite the kernel in terms of Meijer G-functions
\begin{corollary}
We have
\begin{align} \nonumber
K_n(x,y) & =   \int_0^1  G^{0,1}_{1,M+1}
\left(\begin{array}{c} n \\ -\nu_0, \ldots, -\nu_M
\end{array} \Big{|} ux\right)
    G^{M,1}_{M+1,0} \left( \begin{array}{c} -n \\  \nu_0, \ldots, \nu_M \end{array}
    \Big{|} uy \right) \ud u
\\
    & = - \prod_{j=1}^M (n+\nu_j) \int_0^1 P_{n-1}(ux) Q_n(uy) \ud u . \label{KnMeijerG}
\end{align}
\end{corollary}

\begin{proof}
Note that
\begin{equation} \label{u-integral}
    \frac{x^t y^{-s-1}}{s-t} = - \int_0^1 (ux)^t (uy)^{-s-1} \ud u.
    \end{equation}
The kernel \eqref{Knintegral} then is
\begin{align}
        K_n(x,y) =& - \int_0^1  \left( \frac{1}{2\pi i}
        \oint_{\Sigma} \frac{\Gamma(t-n+1)}{\prod_{j=0}^M \Gamma(t+\nu_j+1)} (ux)^t \ud t \right) \nonumber \\
        & \times
    \left( \frac{1}{2\pi i} \int_{-1/2 - i \infty}^{-1/2+i \infty}
        \frac{\prod_{j=0}^M \Gamma(s+\nu_j+1)}{\Gamma(s-n+1)} (uy)^{-s-1} \ud s \right) \ud u.
        \end{align}
By the definition \eqref{def:Meijer} and change of variables $
t\mapsto -t$, $s \mapsto s+1$, both factors in the $u$ integral can
be identified as Meijer G-functions and the first identity in
\eqref{KnMeijerG} follows.

The second identity in \eqref{KnMeijerG} follows from \eqref{QkMeijerG} and \eqref{PnMeijerG}.
\end{proof}

\subsection{Microscopic limit of $K_n$ at the hard edge}
With the help of the contour integral representation \eqref{Knintegral} for $K_n$, we
 derive its scaling limit near the origin (hard edge). The limiting kernels are
denoted by $K^M_{\nu}$, where $\nu$ stands for the collection of
parameters $\nu_1, \ldots, \nu_M$.

\begin{theorem}\label{thm:local limit}
With $\nu_1, \ldots, \nu_M$ being fixed, we have
\[ \lim_{n \to \infty}  \frac{1}{n} K_n \left(\frac{x}{n}, \frac{y}{n} \right) = K^M_{\nu}(x,y), \]
uniformly for $x,y$ in compact subsets of the positive real axis, where
\begin{align}
&K^M_{\nu} (x,y) \nonumber
\\
&=\frac{1}{(2\pi i)^2}
    \int_{-1/2-i\infty}^{-1/2+i\infty} \ud s \int_{\Sigma}  \ud t     \prod_{j=0}^M   \frac{\Gamma(s+\nu_j+1)}{\Gamma(t+\nu_j+ 1)}
        \frac{\sin \pi s}{\sin \pi t} \frac{x^t y^{-s-1}}{s-t}
  \label{def:K(x,y;m)} \\
&=\int_0^1  G^{1,0}_{0,M+1} \left(
\begin{array}{c} - \\-\nu_0, -\nu_1, \ldots, -\nu_M
\end{array}\Big{|} ux \right)
\nonumber\\
& \qquad \qquad \qquad \times
    G^{M,0}_{0,M+1} \left( \begin{array}{c} - \\  \nu_1, \ldots, \nu_M,\nu_0
\end{array}\Big{|}uy
    \right) \ud u, \nonumber
\end{align}
and where $\Sigma$ is a contour starting from $+\infty$ in the upper
half plane and returning to $+\infty$ in the lower half plane which
encircles the positive real axis and $\Re t>-1/2$ for $t\in\Sigma$;
see Figure \ref{fig:contour} for an illustration.
\end{theorem}

\begin{figure}[t]
\centering
\begin{overpic}[scale=0.8]{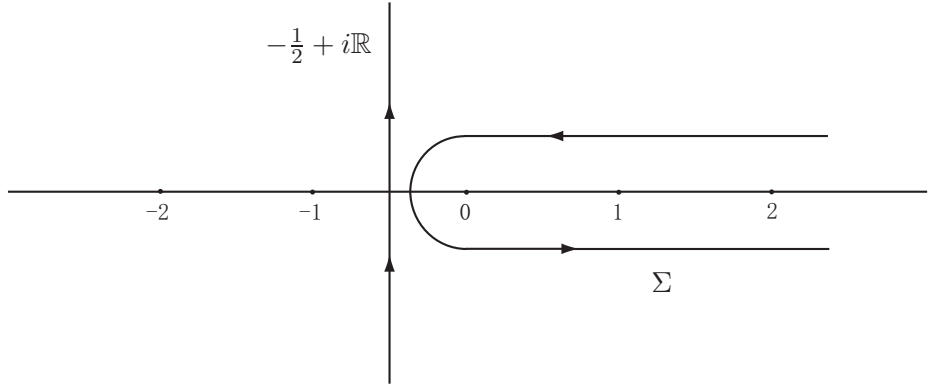}
\put(70,10){$\Sigma$}
\put(28,36){$-\frac{1}{2} + i \mathbb R$}
\end{overpic}
\caption{The two contours of the double integral in \eqref{def:K(x,y;m)}.}
\label{fig:contour}
\end{figure}

\begin{proof}
The reflection formula of the Gamma function says that
\begin{equation}\label{eq:reflection}
\Gamma(t) \Gamma(1-t) = \frac{\pi}{\sin \pi t},
\end{equation}
which means that
\begin{equation}\label{eq:n to -n}
\frac{\Gamma(t-n+1)}{\Gamma(s-n+1)} =\frac{\Gamma(n-s)}{\Gamma(n-t)}
\frac{\sin \pi s}{\sin \pi t}.
\end{equation}

As $n \to \infty$, we have the following ratio asymptotics of Gamma
functions (cf. \cite[formula 5.11.13]{DLMF})
\begin{equation}\label{eq:ratio asy}
\frac{\Gamma(n-s)}{\Gamma(n-t)} = n^{t-s} \left(1+ O(n^{-1})\right),
\end{equation}
which can be easily verified using Stirling's formula.
By modifying the contour $\Sigma$ in \eqref{Knintegral} from a
closed contour around $0, 1, \ldots, n$ to a two sided unbounded
contour as in Figure~\ref{fig:contour} and applying
\eqref{eq:n to -n} and \eqref{eq:ratio asy}, we readily obtain
the first identity in \eqref{def:K(x,y;m)},
provided that we can take the limit inside of the integral.

The $t$-integral in \eqref{def:K(x,y;m)} converges since $\Gamma(t+\nu_j+1)$
increases if we go to infinity along $\Sigma$ and
\[ |\sin \pi t | \geq |\sinh \pi \Im t |. \]
Also the $s$ integral converges since
\[|\mathop{\Gamma\/}\nolimits\!\left(x+iy\right)|\sim\sqrt{2\pi}|y|^{{x-(1/2)}}e%
^{{-\pi|y|/2}},\] as $y\to \pm \infty$ for bounded real value of
$x$; see \cite[formula 5.11.9]{DLMF}. Therefore, $\Gamma(s+\nu_j+1)$
tends to $0$ at an exponential rate if $|s| \to \infty$ with $\Re s
= -1/2$. We can then indeed justify the interchange of limit and
integrals for every $M$ by the dominated convergence theorem.

By \eqref{eq:reflection}, we see \[\frac{\sin \pi s}{\sin \pi t} =
\frac{\Gamma(1+t)\Gamma(-t)}{\Gamma(1+s) \Gamma(-s)},\] and using
the trick \eqref{u-integral} as in the proof of Proposition
\ref{prop:Knintegral}, we obtain
\begin{multline*}
    K_\nu^M(x,y) = - \int_0^1 \left( \frac{1}{2\pi i} \int_{\Sigma} \frac{\Gamma(-t)}{\prod_{j=1}^M \Gamma(t+\nu_j+1)} (ux)^t \ud t \right) \\
    \times \left( \frac{1}{2\pi i} \int_{-1/2 - i \infty}^{-1/2 + i \infty} \frac{\prod_{j=1}^M \Gamma(s + \nu_j+1)}{\Gamma(-s)} (uy)^{-s-1} \ud s \right)
    \ud u.
\end{multline*}
The change of variables $t \mapsto -t$ and $s \mapsto s+1$ takes both integrals into the form \eqref{def:Meijer}
of a Meijer G-function, and the second identity in  \eqref{def:K(x,y;m)} follows.
\end{proof}

It is known that the limiting mean distribution of the squared
singular values for the products of $M$ Ginibre matrices blows up
with a rate $x^{-M/(M+1)}$ near the origin (see \cite{BJLNS,PZ}).
Extending the notion of universality at the hard edge, we are led to
the expectation that the kernels described in Theorem \ref{thm:local
limit} should appear in more general situations of the products of
independent complex random matrices, and possibly in other models of
random matrix theory.

\subsection{Special case $M=1$}
Let's now take a closer look at the limiting kernels
$K_{\nu}^M(x,y)$ for special values of $M$. If $M=1$ and $\nu_1 =
\nu$, one has since $\nu_0 = 0$ (we drop the superscript $M=1$)
\[
K_{\nu} (x,y)= \int_0^1 G^{1,0}_{0,2} \left(   \begin{array}{c} - \\0,
-\nu
\end{array}
    \Big{|} ux \right)
G^{1,0}_{0,2} \left(  \begin{array}{c} - \\  \nu, 0
\end{array}
    \Big{|} uy \right) \ud u.
\]
Since
\begin{align*}
G^{1,0}_{0,2} \left(   \begin{array}{c} - \\0, -\nu
\end{array} \Big{|} ux\right)&=(ux)^{-\nu/2}J_{\nu}(2\sqrt{ux}), \\
G^{1,0}_{0,2} \left(   \begin{array}{c} - \\\nu, 0
\end{array} \Big{|} uy \right)&=(uy)^{\nu/2}J_{\nu}(2\sqrt{uy}),
\end{align*}
where $J_\nu$ denotes the Bessel function of the first kind of order
$\nu$ (see \cite[formula 10.9.23]{DLMF}), it then follows that
\begin{align*}
K_\nu(x,y) & =\left(\frac{y}{x}\right)^{\nu/2}\int_0^1
J_{\nu}(2\sqrt{ux})J_{\nu}(2\sqrt{uy}) \ud u \\
    & = 4   \left(\frac{y}{x}\right)^{\nu/2} K^{\rm Bes,\nu}(4x,4y),
    \end{align*}
where
\[ K^{\rm Bes,\nu}(x,y) = \frac{J_\nu(\sqrt x)\sqrt y J'_\nu
(\sqrt y)-\sqrt x J_\nu'(\sqrt x)J_\nu(\sqrt y)}{2(x-y)},\quad
\nu>-1,
\]
is the Bessel kernel of order $\nu$ that appears as the scaling limit of the Laguerre
or Jacobi unitary ensembles at the hard edge \cite{TW94}, as
expected.

\subsection{Special case $M=2$}
If $M=2$, one has from \eqref{def:K(x,y;m)} that (we  drop the superscript $M=2$)
\begin{equation}\label{eq:K m=2}
K_{\nu_1,\nu_2}(x,y)= \int_0^1 G^{1,0}_{0,3} \left(  \begin{array}{c} - \\0,
-\nu_1,-\nu_2
\end{array}\Big{|} ux
\right) G^{2,0}_{0,3} \left( \begin{array}{c} - \\  \nu_1, \nu_2, 0
\end{array}
    \Big{|} uy \right) \ud u.
\end{equation}
It is interesting that these kernels appeared earlier in another random
matrix model, namely in the Cauchy two-matrix
model with linear potentials, see \cite{BeGeSz09,BeGeSz}.

The Cauchy two matrix model is defined by the probability measure
\begin{equation*}
\frac{1}{\mathcal {Z}_n}\frac{\det (M_1)^a \det (M_2)^b
e^{-\textrm{Tr}\left(V_1(M_1)+V_2(M_2)\right)}}{\det (M_1+M_2)^n}\ud
M_1 \ud M_2, \quad a,b>-1, \, a + b > -1,
\end{equation*}
defined on the space of two $n\times n$ positive semidefinite Hermitian
matrices $M_1$ and $M_2$, with two scalar potentials $V_1, V_2$
defined on the positive real axis that grow sufficiently
fast as $x\to +\infty$.

The eigenvalues of $M_1$ and $M_2$ form a determinantal point process with a
correlation kernel which is
defined in terms of the Cauchy biorthogonal polynomials
\cite{BeGeSz10} $p_l(x)$ and $q_m(y)$ satisfying
\begin{equation*}
\int_0^{\infty}\int_0^{\infty}\frac{x^ay^b
e^{-V_1(x)-V_2(y)}}{x+y}p_l(x)q_m(y)\ud x \ud y=\delta_{l,m}.
\end{equation*}

For the linear case $V_1(x)=x$ and $V_2(y)=y$, it was established in
\cite[Theorem 2.2]{BeGeSz} that the correlation kernel for the eigenvalues
of $M_1$ has a scaling limit at the origin given by
\begin{align}
\int_0^1 G^{1,0}_{0,3} \left(   \begin{array}{c} - \\a, 0,-b
\end{array}\Big{|} ux
\right) G^{2,0}_{0,3} \left(\begin{array}{c} - \\  b, 0, -a
\end{array}
    \Big{|} uy \right) \ud u. \label{kernelsCauchy2MM}
\end{align}
This is slightly different from \eqref{eq:K m=2}, since we cannot
freely permute the parameters $\nu_1, \nu_2, 0$ in \eqref{eq:K m=2}.

However, from \eqref{eq:multiply} we see that
\begin{align*}
 G^{1,0}_{0,3} \left(  \begin{array}{c} - \\a, 0, - b
\end{array}\Big{|} ux
\right) & =
(ux)^a G^{1,0}_{0,3} \left(  \begin{array}{c} - \\ 0, -a,-b-a
\end{array}\Big{|} ux
\right), \\
G^{2,0}_{0,3} \left( \begin{array}{c} - \\ b,0,-a
\end{array}
    \Big{|} uy \right) & =
(uy)^{-a}G^{2,0}_{0,3} \left( \begin{array}{c} - \\  b+a, a, 0
\end{array}
    \Big{|} uy \right).
\end{align*}
Hence,
\begin{equation*}
\int_0^1 G^{1,0}_{0,3} \left(   \begin{array}{c} - \\a, 0,-b
\end{array}\Big{|} ux
\right) G^{2,0}_{0,3} \left(\begin{array}{c} - \\  b, 0, -a
\end{array}
    \Big{|} uy \right) \ud u
        =\left(\frac{x}{y}\right)^a K_{a+b,a}(x,y).
\end{equation*}

The prefactor $\left(\frac{x}{y}\right)^a$ is irrelevant in a kernel
for a determinantal point process as it does not change the determinants
that give the point correlations. Therefore we see that
the limiting kernels \eqref{kernelsCauchy2MM}
in the Cauchy two matrix models are the same kernels
as the limiting kernels for squared singular values of
products of two complex Ginibre matrices. This supports our
conjecture that the kernels \eqref{def:K(x,y;m)} have a universal
character and appear in a wider context.

\subsection{Integrable form of the limiting kernels}
An integral operator with kernel $K(x,y)$ is called integrable if
\[
K(x,y)=\frac{\sum_{i=1}^n f_i(x)g_i(y)}{x-y}, \qquad \text{with} \quad
\sum_{i=1}^nf_i(x)g_i(x)=0,
\]
for some $n \in \{2, 3, \ldots \}$, and certain  functions $f_i$ and $g_i$. Integral operators of this form
benefit from the fact that there is a Riemann-Hilbert setting for the study of
the associated resolvent kernels, determinants, etc.; see
\cite{IIKS}. The kernels of standard universality classes (sine,
Airy, Bessel) encountered in random matrix theory all belong to the class of
integrable operators. The representation \eqref{KnandRHP} of $K_n$ in terms of the
solution of a Riemann-Hilbert problem is also of the integrable form.

 We conclude this paper by giving the integrable
form of the limiting kernels derived in Theorem \ref{thm:local
limit}. Our argument follows \cite[Section 5]{BeGeSz}, where this
was shown for the case $M=2$.

\begin{proposition}\label{prop:integrable}
With $K_{\nu}^M(x,y)$ defined in \eqref{def:K(x,y;m)}, we have
\begin{equation}\label{eq:inte form}
K_\nu^M(x,y)=\frac{\mathcal {B}\left( G^{1,0}_{0,M+1} \left(
\begin{array}{c} - \\-\nu_0, -\nu_1, \ldots, -\nu_M
\end{array} \Big{|} x \right),G^{M,0}_{0,M+1} \left(\begin{array}{c} - \\  \nu_1, \ldots, \nu_M,\nu_0  \end{array}
    \Big{|}y \right)\right)}{x-y},
\end{equation}
where $\mathcal {B}(\cdot,\cdot)$ is a bilinear operator defined by
\begin{align}\label{def:bilinear}
\mathcal {B}\left(f(x),g(y)\right)=(-1)^{M+1}\sum_{j=0}^{M}(-1)^{j}
\left(\Delta_x\right)^jf(x)\left(\sum_{i=0}^{M-j}a_{i+j}\left(\Delta_y\right)^ig(y)
\right),
\end{align}
with $\Delta_x=x\frac{\ud}{\ud x}$ and $\Delta_y=y\frac{\ud}{\ud
y}$. The constants $a_i$ in \eqref{def:bilinear} are determined by
\begin{equation}\label{def:ai}
\prod_{i=1}^M(x-\nu_i)=\sum_{i=0}^M a_i x^i,
\end{equation}
that is,
\begin{equation}\label{eq:elementary}
a_i=(-1)^ie_{M-i}(\nu_1,\ldots,\nu_M)
\end{equation}
with $e_{i}(\nu_1,\ldots,\nu_M)$ being the elementary symmetric
polynomial.
\end{proposition}

The bilinear operator $\mathcal B$ is called a point-split bilinear concomitant in \cite{BeGeSz}.

\begin{proof} We set
\begin{align}\label{def:f}
f(x)&= G^{1,0}_{0,M+1} \left(
\begin{array}{c} - \\-\nu_0, -\nu_1, \ldots, -\nu_M
\end{array}\Big{|} x  \right),
\\
g(y)&=G^{M,0}_{0,M+1} \left(\begin{array}{c} - \\  \nu_1, \ldots,
\nu_M,\nu_0  \end{array}
    \Big{|} y \right). \label{def:g}
\end{align}
By \eqref{def:K(x,y;m)}, our aim is then to evaluate the integral
\begin{equation} \label{eq:Knuintegral}
K_\nu^M(x,y) = \int_0^1 f(tx)g(ty)\ud t.
\end{equation}

Note that the Meijer-G function satisfies the differential equation \eqref{eq:diff}.
For $f$ and $g$ given by \eqref{def:f} and \eqref{def:g}, this implies that for every $t$,
\begin{align}
g(ty) \prod_{j=0}^M(\Delta_x+\nu_j)f(tx)&=-tx f(tx)g(ty), \label{eq:gDf}
\\
f(tx) \prod_{j=0}^M(\Delta_y-\nu_j)g(ty)&=(-1)^M ty f(tx) g(ty). \label{eq:fDg}
\end{align}
If $M$ is odd we subtract these two identities, while if $M$ is even
we add them together. Since the arguments in both cases are similar,
we restrict to the case where $M$ is odd.

Subtracting \eqref{eq:gDf} from \eqref{eq:fDg} we obtain
\begin{align}
&(x-y)f(tx)g(ty) \nonumber
\\
&=\frac{1}{t}\left(f(tx)\prod_{j=0}^M(\Delta_y-\nu_j)g(ty)-g(ty)\prod_{j=0}^M(\Delta_x+\nu_j)f(tx)\right)
\nonumber
\\
&=\frac{1}{t}\sum_{i=0}^M
a_i\left(f(tx)(\Delta_y)^{i+1}g(ty)+(-1)^ig(ty)(\Delta_x)^{i+1}f(tx)\right),
\label{eq:difference}
\end{align}
where the constants $a_i$ are defined in \eqref{def:ai} and
\eqref{eq:elementary}. We next observe that
\begin{multline*}
\frac{\partial}{\partial t}\left(\sum_{j=0}^{i}(-1)^{j}
\left(\Delta_x\right)^jf(tx)\left(\Delta_y\right)^{i-j}g(ty)\right) \\
    = \frac{1}{t} \left( f(tx)(\Delta_y)^{i+1}g(ty) + (-1)^i
    g(ty)(\Delta_x)^{i+1}f(tx)\right),
\end{multline*}
which by \eqref{def:bilinear} and \eqref{eq:difference} implies that
\begin{align} \label{eq:difference2}
    (x-y) f(tx) g(ty) =
        \frac{\partial}{\partial t} \mathcal B(f(tx), g(ty)).
        \end{align}
Using \eqref{eq:difference2} in \eqref{eq:Knuintegral} we find
\[ (x-y) K_{\nu}^M(x,y) = \mathcal B(f(x), g(y)) - \lim_{t \to 0+} \mathcal B(f(tx), g(ty)). \]
It thus remains to show that
\begin{equation}  \label{show:limit}
    \lim_{t \to 0+} \mathcal B(f(tx), g(ty)) = 0,
    \end{equation}
and to do this we need to understand the behavior of $f$ and $g$ at the origin.

First of all, we have by \cite[formula 16.18.1]{DLMF}) and \eqref{def:f} that $f$ is a hypergeometric
function
\begin{align} \label{eq:asy f}
f(x) = \frac{1}{\prod_{j=1}^M \Gamma(1-\nu_j)}
\mathop{{{}_{{0}}F_{{M}}}\/}\nolimits\!\left({-\atop 1-\nu_1,
\dots,1-\nu_M}  \Big{|} -x\right),
\end{align}
so that $f$ is analytic at the origin. Next by \eqref{def:g}, the
definition of \eqref{def:Meijer}, and the properties of the Mellin
transform (see e.g.\ \cite{FlGoDu}), we find
\[ \int_0^{\infty} (\Delta_y)^i  g(y) y^{s-1} \ud y = (-s)^i \frac{ \prod_{j=1}^M \Gamma(s+\nu_j)}{\Gamma(1-s)}. \]
Then it follows in the same way as we obtained \eqref{w0:at0} that
\begin{align}\label{eq:asy g}
 (\Delta y)^i g(y)= O(y^{\alpha} (\log y)^{r-1}) \qquad \text{ as } y \to 0+,
\end{align}
with $\alpha = \min(\nu_1, \ldots, \nu_M) > -1$ and $r = \# \{ j \mid \nu_j = \alpha\}$.

Now we look at the $j=0$ term in \eqref{def:bilinear} which is
\begin{align*}
    f(x) \sum_{i=0}^M a_i (\Delta_y)^i g(y) & = f(x) \prod_{i=0}^M (\Delta_y - \nu_i) g(y)
    = -f(x) y g(y),
    \end{align*}
where in the last step we used \eqref{eq:fDg} with $t=1$.
Replacing $x \mapsto tx$, $y \mapsto ty$, we find by \eqref{eq:asy f} and  \eqref{eq:asy g}
that the limit is $0$ as $t \to 0+$.
For $j \geq 1$ we have
\[ (\Delta_x)^j f(x) = O(x) \qquad \text{ as } x \to 0, \]
and then it follows from \eqref{eq:asy g} that the terms in
\eqref{def:bilinear} with $j \geq 1$ are all $O(x) O(y^{\alpha}
(\log y)^{r-1)}$ as $x, y \to 0+$. Replacing $x \mapsto tx$, $y
\mapsto ty$, we then find that these terms tend to $0$ as well as $t
\to 0+$. This proves \eqref{show:limit} and it completes the proof
of Proposition \ref{prop:integrable}.
\end{proof}

\appendix

\section{The Meijer G-function}
We give a brief introduction to the Meijer G-function in this
appendix. By definition, the Meijer G-function is given by the
following contour integral in the complex plane:
\begin{multline}\label{def:Meijer}
G^{m,n}_{p,q}\left({a_1,\ldots,a_p \atop b_1,\ldots,b_q}\Big{|}
z\right)=G^{m,n}_{p,q}\left({\bf{a_p} \atop
\bf{b_q}}\Big{|}z\right)\\
=\frac{1}{2\pi i}\int_\gamma
\frac{\prod_{j=1}^m\Gamma(b_j+u)\prod_{j=1}^n\Gamma(1-a_j-u)}
{\prod_{j=m+1}^q\Gamma(1-b_j-u)\prod_{j=n+1}^p\Gamma(a_j+u)}z^{-u}
\ud u,
\end{multline}
where $\Gamma$ denotes the usual gamma function and the branch cut
of $z^{-u}$ is taken along the negative real axis. It is also
assumed that
\begin{itemize}
  \item $0\leq m\leq q$ and $0\leq n \leq p$, where $m,n,p$ and $q$
  are integer numbers;
  \item The real or complex parameters $a_1,\ldots,a_p$ and
  $b_1,\ldots,b_q$ satisfy the conditions
  \begin{equation*}
  a_k-b_j \neq 1,2,3, \ldots, \quad \textrm{for $k=1,2,\ldots,n$ and $j=1,2,\ldots,m$,}
  \end{equation*}
  i.e., none of the poles of $\Gamma(b_j+u)$, $j=1,2,\ldots,m$ coincides
  with any poles of $\Gamma(1-a_k-u)$, $k=1,2,\ldots,n$.
\end{itemize}
The contour $\gamma$ is chosen in such a way that all the poles of
$\Gamma(b_j+u)$, $j=1,\ldots,m$ are on the left of the path, while
all the poles of $\Gamma(1-a_k-u)$, $k=1,\ldots,n$ are on the right,
which is usually taken to go from $-i\infty$ to $i\infty$. In
particular, it can be a loop starting and ending at $+\infty$ if
$p>q$, or a loop beginning and ending at $-\infty$ if $p<q$. Most of
the known special functions can be viewed as special cases of the
Meijer G-functions, we refer to \cite{Luke,DLMF} for more details.
We end this appendix with several formulas used in this paper.
\begin{itemize}
\item From the definition \eqref{def:Meijer}, it is easily seen that
   \begin{equation}\label{eq:multiply}
   z^{\rho}G^{m,n}_{p,q}\left({\bf{a_p}\atop
        \bf{b_q}}\Big{|}z\right)=G^{m,n}_{p,q}\left({\bf{a_p}+\rho \atop
        \bf{b_q}+\rho}\Big{|}z\right).
   \end{equation}

  \item The Meijer G-function $G^{m,n}_{p,q}\left({\bf{a_p} \atop
\bf{b_q}}\Big{|}z\right)$ satisfies the following linear
differential equation of order $\max(p,q)$:
\begin{multline}\label{eq:diff}
\bigg[(-1)^{p-m-n}z\prod_{j=1}^{p}\left(z\frac{\ud}{\ud
z}-a_j+1\right)\\
-\prod_{j=1}^{q}\left(z\frac{\ud}{\ud
z}-b_j\right)\bigg]G^{m,n}_{p,q}\left({\bf{a_p} \atop
\bf{b_q}}\Big{|}z\right)=0;
\end{multline}
see \cite[formula 16.21.1]{DLMF}.
\end{itemize}

\section*{Acknowledgements}
We thank  Gernot Akemann, Jesper R.~Ipsen and Mario Kieburg for interesting
discussions and for providing us with an early copy of  the preprint~\cite{AIK}.

The first author is supported by KU Leuven Research Grant  OT/12/073, the Belgian
Interuniversity Attraction Pole P07/18, FWO Flanders projects
G.0641.11 and G.0934.13, and by Grant No.~MTM2011-28952-C02 of the
Spanish Ministry of Science and Innovation.
The second author was a Postdoctoral Fellow of the Fund for Scientific Research - Flanders (Belgium), 
and is also supported by The Program for Professor of Special Appointment (Eastern Scholar) 
at Shanghai Institutions of Higher Learning (No. SHH1411007) and by Grant SGST 12DZ 2272800 from Fudan University.

\end{document}